\documentclass[12pt,reqno]{amsart}
\usepackage{amsmath}
\usepackage{latexsym}
\usepackage{amsfonts}
\usepackage{amssymb}
\usepackage{color}
\usepackage{bbm,dsfont}
\usepackage{graphicx}
\usepackage{hyperref}
\usepackage{enumerate}
\usepackage[foot]{amsaddr}

\newtheorem{proposition}{Proposition}
\newtheorem{theorem}{Theorem}

\newtheorem{corollary}{Corollary}

\theoremstyle{definition}

\newtheorem{example}{Example}
\newtheorem{definition}{Definition}

\newcommand{\real}{\mathbb R} 
\newcommand{\half}{\tfrac{1}{2}} 

\newcommand{\hi}{\mathcal{H}} 
\newcommand{\lh}{\mathcal{L(H)}} 
\newcommand{\kb}[2]{|#1\rangle\langle#2|} 
\newcommand{\no}[1]{\left\|#1\right\|} 
\newcommand{\id}{\mathbbm{1}} 

\newcommand{\va}{\mathbf{a}} 
\newcommand{\vb}{\mathbf{b}} 
\newcommand{\vg}{\mathbf{g}} 
\newcommand{\vx}{\mathbf{x}} 
\newcommand{\vy}{\mathbf{y}} 
\newcommand{\vsigma}{\boldsymbol{\sigma}} 

\newcommand{\joint}{\mathcal{J}} 
\newcommand{\minjoint}{\mathcal{J}_{\mathrm{min}}} 
\newcommand{\maxjoint}{\mathcal{J}_{\mathrm{max}}} 

\newcommand{\Ao}{\mathsf{A}}
\newcommand{\Bo}{\mathsf{B}}
\newcommand{\Co}{\mathsf{C}}
\newcommand{\Fo}{\mathsf{F}}
\newcommand{\Go}{\mathsf{G}}

\newcommand{\Ea}{\mathsf{E}^{1,\va}} 
\newcommand{\Eb}{\mathsf{E}^{1,\vb}} 
\newcommand{\Eaa}{\mathsf{E}^{\alpha,\mathbf{a}}} 
\newcommand{\Ebb}{\mathsf{E}^{\beta,\mathbf{b}}} 
\newcommand{\Gg}{\mathsf{G}^{\gamma,\mathbf{g}}} 

\newcommand{\pp}{\prec_{\mathrm{post}}} 
\newcommand{\ppeq}{\preceq_{\mathrm{post}}} 
\newcommand{\ppsim}{\sim_{\mathrm{post}}} 
\newcommand{\markov}[2]{\mathbf{Markov}(#1 , #2)} 

\newcommand{\defarrow}{\stackrel{\mathrm{def.}}{\Leftrightarrow}}

\begin{document}

\title[]{Post-processing Minimal Joint Observables}
\author{Teiko Heinosaari$^\Diamond$}
\address{$\Diamond$ QTF Centre of Excellence, Department of Physics and Astronomy, University of Turku, Turku 20014, Finland}

\author{Yui Kuramochi$^\Box$}
\address{$\Box$ School of Physics and Astronomy, Sun Yat-Sen University, Zhuhai Campus, Zhuhai 519082, China}


\begin{abstract}
A finite set of quantum observables (positive operator valued measures) is called compatible if these observables are marginals of a some observable, called a joint observable of them.
For a given set of compatible observables, their joint observable is in general not unique and it is desirable to take a minimal joint observable in the post-processing order since a less informative observable disturbs less the system.
We address the question of the minimality of finite-outcome joint observables and prove that any joint observable is lower bounded by a minimal joint observable in the post-processing order.
We also give characterizations of the minimality of a joint observable that can be checked by finite-step algorithms and apply them to the case of non-commuting dichotomic qubit observables.
\end{abstract}

\maketitle

\section{Introduction}

An important and intriguing feature of quantum observables is that two observables may not allow a simultaneous measurement, or any other kind of joint implementation. 
This relation is called incompatibility, and it links interestingly to various other features of quantum theory \cite{HeMiZi16}.

Mathematically speaking, quantum observables are described as positive operator valued measures (POVMs). 
Two observables $\Ao$ and $\Bo$ are compatible if there exists a third observable $\Co$ that gives both of them as marginals. 
In this case, $\Co$ is called their joint observable. 
The set of all joint observables of a compatible pair of observables is convex. 
Therefore, a compatible pair has either a unique joint observable or infinitely many different joint observables. 
If two observables are compatible and at least one of them is sharp, then their joint observable is unique \cite{HeReSt08}.
Other criteria leading to the existence of a unique joint observable have been studied in \cite{HaHePe14,GuCu18}.
These are, in any case, very special situations and in the generic case there is no unique joint observable. 
This raises the question on the possible, physically motivated, hiearchy in the set of all joint observables. 
As a practical problem, one may ask if some joint observables are more preferred than others.  

A possible starting point, adopted in the current investigation, is that one should choose a joint observable that allows least disturbing measurement among all joint observables. 
As explained in Section~\ref{sec:motivation}, this means that we are interested in the post-processing ordering of the set joint observables, and seek minimal elements in that set. 
Our main result is that any joint observable is lower bounded by a minimal joint observable 
(Section~\ref{sec:min}, Theorem~\ref{theo:min} and Corollary~\ref{coro:complete}).
We also give characterizations of minimal joint observables that can be checked by finite-step algorithms for a given joint observable (Section~\ref{sec:characterization}).
We apply these characterizations to two non-commutative dichotomic qubit observables
and obtain a complete characterization of minimality of joint observables in this case
(Theorem~\ref{theo:qubitmin}).

\section{Motivation of the question}\label{sec:motivation}

\subsection{Post-processing minimal observables}

In this work we are going to deal with observables with finite number of outcomes.
Let $\hi$ be a complex Hilbert space. 
We denote by $\lh$ the set of all bounded operators on $\hi$.
An observable is a map $\Ao:\Omega\to\lh$ such that $\Ao(x)\geq 0$ for every $x\in\Omega$ and $\sum_x \Ao(x)=\id$, where $\Omega$ is a finite set of measurement outcomes. 

As an introductory example, suppose our aim is to discriminate two orthogonal pure states (i.e. one-dimensional projections) $P_1=\kb{\psi_1}{\psi_1}$ and $P_2=\kb{\psi_2}{\psi_2}$ of a, say, $4$-dimensional quantum system.
One possibility is to complete the set $\{\psi_1,\psi_2\}$ into an orthonormal basis $\{\psi_1,\psi_2,\psi_3,\psi_4\}$ and then make a measurement in that basis.
The corresponding observable is hence $\Ao(x)=\kb{\psi_x}{\psi_x}$, $x=1,\ldots,4$. 
 This kind of measurement, however, disturbs the system more than is necessary.
Instead of measuring $\Ao$, we can, for instance, perform a measurement of a two-outcome observable $\Bo$, defined as $\Bo(1)=\kb{\psi_1}{\psi_1} + \kb{\psi_3}{\psi_3}$, $\Bo(2)=\kb{\psi_2}{\psi_2} + \kb{\psi_4}{\psi_4}$.
The observable $\Bo$ discriminates the states $P_1$ and $P_2$, but allows a less disturbing measurement than $\Ao$.
In this exemplary case, we can get $\Bo$ from $\Ao$ by grouping the outcomes.  
Namely, we have
\begin{equation*}
\Bo(1) = \Ao(1)+\Ao(3) \, , \quad \Bo(2) = \Ao(2) + \Ao(4) \, .
\end{equation*}
This sort of relabeling of outcomes is a special type of post-processing.

To recall the general definition of post-processing, let $\Omega_1$ and $\Omega_2$ be finite sets.
A map $p \colon \Omega_1 \times \Omega_2 \to \real$ is called 
a \emph{Markov kernel}, or Markov matrix, from $\Omega_2$ to $\Omega_1$
if for each $x \in \Omega_1 $ and each $y \in \Omega_2$,
$p(x,y) \geq 0$
and 
$\sum_{x^\prime \in \Omega_1} p(x^\prime , y) =1.$
The set of Markov kernels from $\Omega_2$ to $\Omega_1$ is 
written as $\markov{\Omega_1}{\Omega_2}$, and it is a compact convex subset of 
the Euclidean space $\real^{\Omega_1 \times \Omega_2}$.
Let $\Ao \colon \Omega_1 \to \lh$ be an observable, let $\Omega_2$ be a finite set,
and let $p \in \markov{\Omega_2}{\Omega_1}$.
We define an observable $p \ast \Ao \colon \Omega_2 \to \lh$
by
\begin{equation*}
p\ast \Ao (x) := \sum_{y \in \Omega_1} p(x,y) \Ao (y) \, .
\end{equation*}
If an observable $\Bo$, with an outcome set $\Omega_2$, can be written as $p \ast \Ao $
for some $p \in \markov{\Omega_2}{\Omega_1}$, then $\Bo$ is called a \emph{post-processing} of $\Ao$, and 
written as $\Bo \ppeq \Ao$.
This relation has been studied e.g. in \cite{MaMu90a,BuDaKePeWe05,Heinonen05}.

The relation $\ppeq$ is reflexive and transitive, hence a preorder.
A preorder induces an equivalence relation $\ppsim$ and 
a partial order in the set of equivalence classes. 
We write $\Ao \pp \Bo$ if $\Ao \ppeq \Bo$ holds but
$\Bo \ppeq \Ao$ does not hold.
We will say that an observable $\Ao$ is maximal/minimal/greatest/least in a subset $X$ if the corresponding equivalence class $[\Ao]$ has that property in the respective subset of equivalence classes.  
For instance, we say that an observable $\Ao$ is \emph{post-processing minimal  (in $X$)} if the following implication holds for every observable $\Bo$ $(\in X)$:
\begin{equation*}
 \Bo \ppeq \Ao \quad \Rightarrow \quad \Ao \ppeq \Bo \, .
\end{equation*}
It is known that an observable $\Ao$ is post-processing minimal in the set of all observables if and only if each of its operator is a multiple of the identity operator $\id$, while $\Ao$ is post-processing maximal in the set of all observables if and only if each of its nonzero operator is rank-1 \cite{MaMu90a}.

It has been shown in \cite{HeMi13} that for two observables $\Ao$ and $\Bo$, the relation $\Bo \ppeq \Ao$ holds if and only if the disturbance related to $\Bo$ is smaller than or equal to the disturbance related to $\Ao$.
The disturbance related to an observable $\Ao$ refers to the set of all quantum channels that arise in some measurement of $\Ao$.
As we may want to perform subsequent measurements, it is desirable to disturb the initial state as little as possible. 
We thereby take the following as a guiding principle:
\begin{quote}
\emph{Whenever we need to choose an observable that has certain property (e.g. enables discrimination of some states), we should choose it to be minimal in the post-processing preorder in the set of all observables with that property.}
\end{quote}

In a typical case, this guiding principle does not lead to a unique choice.
For instance, in the starting example we can also choose $\Bo'(1)=\kb{\psi_1}{\psi_1} + \kb{\psi_4}{\psi_4}$, $\Bo'(2)=\kb{\psi_2}{\psi_2} + \kb{\psi_3}{\psi_3}$, or any convex combination of $\Bo$ and $\Bo'$.
In this specific example any two-outcome discriminating observable is post-processing minimal.
To see this, we first note that a two-outcome observable $\Co$ discriminates $\psi_1$ and $\psi_2$
if and only if $\Co(i) \geq P_i ,$
$i=1,2 .$
Let $\Co$ and $\Co'$ be two-outcome observables that discriminate 
$\psi_1$ and $\psi_2 $
and suppose that $\Co \ppeq \Co'$.
We take a Markov kernel $p$ such that $\Co = p \ast \Co'$.
Then 
\begin{gather*}
	P_1 = P_1 \Co(1) =  \sum_{i=1,2} p(1,i) P_1 \Co'(i) = p(1,1) P_1 
	\\
	P_2 = P_2 \Co(2) =  \sum_{i=1,2} p(2,i) P_2 \Co'(i) = p(2,2) P_2
\end{gather*}
Therefore $p(1,1)=p(2,2)=1$ and hence $p(1,2)=p(2,1)=0$.
Thus $\Co = \Co'$.
This implies the post-processing minimality of any two-outcome discriminating observable.

\subsection{Minimal joint observables}

In our investigation we consider a situation where there are several tasks that we want to perform. 
We assume that observables $\Ao_1,\ldots,\Ao_n$ for the individual tasks have already been chosen and that they are compatible.
We recall that observables $\Ao_1,\ldots,\Ao_n$ are, by definition, compatible if there exists an observable $\Co$ such that $\Ao_\ell \ppeq \Co$ for all $\ell=1,\ldots,n$.
Clearly, measuring $\Co$ is then enough to simulate measurements of all $\Ao_1,\ldots,\Ao_n$.

According to our guiding principle, we want to choose a post-processing minimal observable among all observables $\Co$ that satisfy $\Ao_\ell \ppeq \Co$ for all $\ell=1,\ldots,n$. 
Our first observation, based on \cite{AlCaHeTo09}, is that if such $\Co$ exists, then there also exists an observable $\Go$ with the outcome set $\Omega_1 \times \cdots \times \Omega_n$ such that $\Go \ppeq \Co$ and each $\Ao_\ell$ is the $\ell$th marginal of $\Go$, i.e., 
\begin{align*}
& \sum_{x_2,\ldots,x_n} \Go(x_1,x_2,\ldots,x_n) = \Ao_1(x_1) \\
& \qquad\qquad \vdots \\
& \sum_{x_1,\ldots,x_{n-1}} \Go(x_1,x_2,\ldots,x_n) = \Ao_n(x_n) \\
\end{align*}
This kind of observable is called a \emph{joint observable} of $\Ao_1,\ldots,\Ao_n$ \cite{LaPu97}.

To verify the previous claims, we first recall that $\Ao_\ell \ppeq \Co$ means that there exists a Markov kernel $p_\ell$ such that
\begin{equation*}
\Ao_\ell(x) = \sum_{y} p_\ell(x,y) \Co(y) \, .
\end{equation*}
We define
\begin{equation}
p((x_1,\ldots,x_n),y) := \prod_{\ell=1}^n p_\ell (x_\ell,y)
\end{equation}
and set
\begin{equation}\label{eq:GfromC}
\Go(x_1,\ldots,x_n) := \sum_{y} p((x_1,\ldots,x_n),y) \Co(y) \, .
\end{equation}
It is straightforward to verify that $\Go$ is a joint observable of $\Ao_1,\ldots,\Ao_n$.
Further,  $p$ is a Markov kernel and, thus, \eqref{eq:GfromC} means that $\Go \ppeq \Co$.

We conclude that when we search for minimal observables $\Co$ satisfying $\Ao_\ell \ppeq \Co$ for all $\ell=1,\ldots,n$, we can limit our search for joint observables of $\Ao_1,\ldots,\Ao_n$.
Namely, any $\Co$ with the required property is either post-processing equivalent to some joint observable, or strictly greater than some joint observable. 

For a finite set of observables $\{\Ao_\ell\}_{\ell =1}^n$, we denote by $\joint(\{\Ao_\ell\}_{\ell =1}^n)$ the set of all their joint observables. 
The following concept will be our main focus.

\begin{definition}\label{def:min}
Let $\{\Ao_\ell\}_{\ell =1}^n$ be a set of compatible observables.
Their joint observable $\Go$ is a \emph{minimal joint observable} if $\Go$ is post-processing minimal in $\joint(\{\Ao_\ell\}_{\ell =1}^n)$.
The set of minimal observables of $\{\Ao_\ell\}_{\ell =1}^n$ is denoted by
$\minjoint(\{\Ao_\ell\}_{\ell =1}^n)$.
\end{definition}

Since the post-processing relation is a preorder rather than a partial order, it is often convenient to work with the equivalence classes of observables, and in that case the induced post-processing relation is a partial order. 
We denote by $\minjoint ( \{\Ao_\ell \}_{\ell =1}^n)/\ppsim$ the partially ordered set of equivalence classes of minimal joint observables.

\section{Order structure of the set of joint observables}\label{sec:min}

In this section $\{\Ao_\ell\}_{\ell =1}^n$ is a fixed set of compatible observables.
The condition of minimality for a joint observable is 
independent of the choice of the post-processing representatives
of the marginal observables as shown in the following proposition.
\begin{proposition}
\label{prop:min_independence}
Let $\{\Ao_\ell\}_{\ell =1}^n$ be a set of compatible observables
and let $\{\Ao_\ell^\prime \}_{\ell =1}^n$ be observables 
such that $\Ao_\ell \ppsim \Ao_\ell^\prime$
for all $\ell .$
Then there exists a bijection
\[
f \colon \minjoint ( \{\Ao_\ell \}_{\ell =1}^n)/\ppsim 
\to 
\minjoint (\{\Ao_\ell^\prime \}_{\ell =1}^n)/\ppsim
\]
such that 
$[\Go] \ppsim  f([ \Go  ])   $
for all $\Go \in \minjoint (\{\Ao_\ell \}_{\ell =1}^n) .$
\end{proposition}
\begin{proof}
Let $\Go \in \minjoint (\{\Ao_\ell \}_{\ell =1}^n) .$
Since $\Ao_\ell^\prime \ppsim \Ao_\ell \ppeq \Go ,$
there exists a joint observable $\Go^\prime \in \joint ( \{\Ao_\ell^\prime \}_{\ell =1}^n)$
such that $\Go^\prime \ppeq \Go .$
Take an arbitrary $\Go_0^\prime \in  \joint ( \{\Ao_\ell^\prime \}_{\ell =1}^n)$
satisfying $\Go_0^\prime \ppeq  \Go^\prime .$
Since $\Ao_\ell \ppeq \Go_0^\prime$ for all $\ell ,$
there exists a joint observable 
$\Go_0 \in \joint ( \{\Ao_\ell \}_{\ell =1}^n)$ such that
$\Go_0 \ppeq \Go_0^\prime .$
Then we have $\Go_0 \ppeq \Go_0^\prime \ppeq \Go^\prime \ppeq \Go $
and the minimality of $\Go$ implies that
all of the joint observables $\Go_0 ,$ $ \Go_0^\prime$ and $\Go^\prime$
are post-processing equivalent to $\Go .$
This shows that $\Go^\prime$ is minimal and post-processing equivalent to $\Go .$
Hence there exists a mapping
$f \colon \minjoint ( \{\Ao_\ell \}_{\ell =1}^n)/\ppsim \to \minjoint (\{\Ao_\ell^\prime \}_{\ell =1}^n)/\ppsim$
such that 
$[\Go] \ppsim  f([ \Go  ])   $
for all $\Go \in \minjoint (\{\Ao_\ell \}_{\ell =1}^n) .$
$f$ is apparently injective.
By interchanging $\{\Ao_\ell \}_{\ell =1}^n$ with 
$\{\Ao_\ell^\prime \}_{\ell =1}^n $ in the above discussion,
we can also conclude the existence of an injection 
$
g \colon \minjoint ( \{\Ao_\ell^\prime \}_{\ell =1}^n)/\ppsim 
\to 
\minjoint (\{\Ao_\ell \}_{\ell =1}^n)/\ppsim
$
such that 
$[\Go^\prime] \ppsim  g([ \Go^\prime  ])   $
for all $\Go^\prime \in \minjoint (\{\Ao_\ell^\prime \}_{\ell =1}^n) .$
Then $g $ is the inverse map of $f$ and hence $f$ is a bijection.
\end{proof}

The following theorem states that any post-processing monotone net 
of joint observables has its supremum or infinimum that is also a joint observable.
\begin{theorem}
\label{theo:min}
Let $\{\Ao_\ell\}_{\ell =1}^n$ be a finite set of compatible observables
and let $([\Go_\lambda])_{\lambda \in \Lambda}$ be a 
net in $\joint (\{\Ao_\ell\}_{\ell =1}^n) / \ppsim .$
\begin{enumerate}[1.]
\item
If $([\Go_\lambda])_{\lambda \in \Lambda}$ is monotonically increasing,
then there exists a joint observable $\widetilde{\Go} \in \joint (\{\Ao_\ell\}_{\ell =1}^n)$
such that $[\widetilde{\Go}]$ is an upper bound of the net 
$([\Go_\lambda])_{\lambda \in \Lambda}$
and for any observable $\Bo ,$
$\Go_\lambda \ppeq \Bo$ $(\forall \lambda \in \Lambda)$
implies 
$\widetilde{\Go} \ppeq \Bo  .$
\item
If $([\Go_\lambda])_{\lambda \in \Lambda}$ is monotonically decreasing,
then there exists a joint observable $\widetilde{\Go} \in \joint (\{\Ao_\ell\}_{\ell =1}^n)$
such that $[\widetilde{\Go}]$ is a lower bound of the net 
$([\Go_\lambda])_{\lambda \in \Lambda}$
and for any observable $\Bo ,$
$\Bo \ppeq \Go_\lambda $ $(\forall \lambda \in \Lambda)$
implies 
$\Bo \ppeq \widetilde{\Go}  .$
\end{enumerate}
\end{theorem}
\begin{proof}
Let $\Omega_\ell$ be the outcome set of $\Ao_\ell$
and let $\tilde{\Omega}:= \Omega_1 \times \dots \times \Omega_n .$
\begin{enumerate}[1.]
\item
By assumption, for $\lambda, \lambda^\prime \in \Lambda$
satisfying $\lambda \leq \lambda^\prime$
there exists $p^{\lambda | \lambda^\prime} \in \markov{\tilde{\Omega}}{\tilde{\Omega}}$
such that $\Go_\lambda = p^{\lambda | \lambda^\prime}  \ast \Go_{\lambda^\prime} .$
For each $\lambda , \lambda^\prime \in \Lambda$
we define a Markov kernel
$q^{\lambda | \lambda^\prime} \in \markov{\tilde{\Omega}}{\tilde{\Omega}}$
by
\begin{equation}
	q^{\lambda | \lambda^\prime} (\vx , \vx^\prime)
	:=
	\begin{cases}
		p^{\lambda | \lambda^\prime} (\vx , \vx^\prime)
		& \text{if $\lambda \leq \lambda^\prime ,$} \\
		\delta_{\vx , \vx^\prime}
		& \text{otherwise.}
	\end{cases}
	\notag
\end{equation}
By Tychonoff's theorem, the set 
$\markov{\tilde{\Omega}}{\tilde{\Omega}}^\Lambda \times \joint (\{\Ao_\ell\}_{\ell =1}^n)$
is compact in the product topology.
Hence there exists a subnet $(\Go_{\lambda^\prime (i)})_{i \in I}$
such that
$q^{\lambda | \lambda^\prime (i)}$ converges to some 
$q^\lambda \in \markov{\tilde{\Omega}}{\tilde{\Omega}}$
for each $\lambda \in \Lambda$
and 
$\Go_{\lambda^\prime (i)}$ converges to some 
$\widetilde{\Go} \in \joint (\{\Ao_\ell\}_{\ell =1}^n) .$
For each $\lambda \in \Lambda$ and each
$i \in I$ satisfying $\lambda \leq \lambda^\prime (i) ,$
we have 
$
\Go_\lambda 
= p^{\lambda | \lambda^\prime (i)} \ast \Go_{\lambda^\prime (i)}
= q^{\lambda | \lambda^\prime (i)} \ast \Go_{\lambda^\prime (i)} .
$
Thus 
\begin{align*}
	&\left\| 
	\Go_\lambda (\vx)  
	- \sum_{\vx^\prime \in \tilde{\Omega}}
	q^\lambda (\vx , \vx^\prime) \widetilde{\Go} (\vx^\prime)
	\right\|
	\\
	&=
	\left\| 
	\sum_{\vx^\prime \in \tilde{\Omega}}
	q^{\lambda| \lambda^\prime (i)} (\vx , \vx^\prime) \Go_{\lambda^\prime (i)} (\vx^\prime)
	- \sum_{\vx^\prime \in \tilde{\Omega}}
	q^\lambda (\vx , \vx^\prime) \widetilde{\Go} (\vx^\prime)
	\right\|
	\\
	&\leq
	\sum_{\vx^\prime \in \tilde{\Omega}}
	|
	q^{\lambda| \lambda^\prime (i)} (\vx , \vx^\prime) 
	-
	q^\lambda (\vx , \vx^\prime)
	|
	\left\|
	\Go_{\lambda^\prime (i)} (\vx^\prime)
	\right\|
	+
	\sum_{\vx^\prime \in \tilde{\Omega}}
	q^\lambda (\vx , \vx^\prime)
	\left\|
	\Go_{\lambda^\prime (i)} (\vx^\prime)
	-
	\widetilde{\Go} (\vx^\prime)
	\right\|
	\\
	&\leq
	\sum_{\vx^\prime \in \tilde{\Omega}}
	|
	q^{\lambda| \lambda^\prime (i)} (\vx , \vx^\prime) 
	-
	q^\lambda (\vx , \vx^\prime)
	|
	+
	\sum_{\vx^\prime \in \tilde{\Omega}}
	q^\lambda (\vx , \vx^\prime)
	\left\|
	\Go_{\lambda^\prime (i)} (\vx^\prime)
	-
	\widetilde{\Go} (\vx^\prime)
	\right\|
	\\
	&\to 0,
\end{align*}
where $\| \cdot \|$ denotes the operator norm.
This implies $\Go_\lambda = q^\lambda \ast \widetilde{\Go} \ppeq \widetilde{\Go} .$
Hence $[\widetilde{\Go}]$ is an upper bound of the net $([\Go_\lambda])_{\lambda \in \Lambda} .$

Now we take an observable $(\Bo (y) )_{ y \in \Omega^\prime} $ satisfying
$\Go_\lambda \ppeq \Bo$ for all $\lambda \in \Lambda.$
Then for each $\lambda \in \Lambda $ there exists 
$r^\lambda \in \markov{\tilde{\Omega}}{\Omega^\prime}$
such that $\Go_\lambda = r^\lambda \ast \Bo . $
Since $\markov{\tilde{\Omega}}{\Omega^\prime}$ is compact, 
there exists a subnet $(r^{\lambda^{\prime\prime} (j)})_{j \in J}$
of $(r^{\lambda^{\prime} (i)})_{i \in I}$ such that
$r^{\lambda^{\prime\prime} (j)}$ converges to some
$\tilde{r} \in \markov{\tilde{\Omega}}{\Omega^\prime} .$
Thus
\begin{align*}
	\widetilde{\Go} (\vx)
	&=
	\lim_{j \in J} \Go_{\lambda^{\prime\prime} (j)} (\vx)
	=
	\lim_{j \in J} 
	\sum_{y \in \Omega^\prime}
	r^{\lambda^{\prime\prime} (j)} (\vx , y)
	\Bo (y) 
	=
	\sum_{y \in \Omega^\prime}
	\tilde{r} (\vx , y)
	\Bo (y) ,
\end{align*}
which implies $\widetilde{\Go} = \tilde{r}\ast \Bo \ppeq \Bo .$
\item
For each joint observable $\Go \in \joint (\{\Ao_\ell\}_{\ell =1}^n)$ 
we define
\begin{equation}
K_{\Go} :=
\left\{  p \in \markov{\tilde{\Omega}}{\tilde{\Omega}} \middle|
p \ast \Go \in \joint (\{\Ao_\ell\}_{\ell =1}^n)
\right\} ,
\notag 
\end{equation}
which is a nonempty compact subset of $\markov{\tilde{\Omega}}{\tilde{\Omega}} .$ 
By assumption, for $\lambda^\prime \geq \lambda $ there exists a Markov kernel
$p^{\lambda^\prime | \lambda} \in K_{\Go_\lambda} $
such that $\Go_{\lambda^\prime} = p^{\lambda^\prime | \lambda} \ast \Go_\lambda .$
For each $\lambda , \lambda^\prime \in \Lambda ,$
we define $q^{\lambda^\prime | \lambda } \in K_{\Go_\lambda}$ by
\[
	q^{\lambda^\prime | \lambda} (\vx^\prime , \vx)
	:=
	\begin{cases}
		p^{\lambda^\prime | \lambda} (\vx^\prime , \vx)
		& \text{if $  \lambda \leq \lambda^\prime ,$} \\
		\delta_{\vx^\prime , \vx}
		& \text{otherwise.}
	\end{cases}
\]
By Tychonoff's theorem, the set 
$\prod_{\lambda \in \Lambda} K_{\Go_\lambda}$ is compact in the product topology,
and hence there exists a subnet 
$(\Go_{\lambda^\prime (i)})_{i \in I}$ such that 
$(q^{\lambda^\prime (i)|\lambda})_{i \in I}$
converges to some $q^{\lambda} \in K_{\Go_\lambda}$
for every $\lambda \in \Lambda .$
For each $\lambda \in \Lambda $
and each $ i \in I$ satisfying $\lambda^\prime (i) \geq \lambda ,$
we have
\begin{align*}
	\Go_{\lambda^\prime (i)} (\vx^\prime)
	&=
	\sum_{\vx \in \tilde{\Omega}}
	p^{\lambda^\prime (i) | \lambda} (\vx^\prime  , \vx)
	\Go_\lambda (\vx)
	=
	\sum_{\vx \in \tilde{\Omega}}
	q^{\lambda^\prime (i) | \lambda} (\vx^\prime  , \vx)
	\Go_\lambda (\vx) \\
	& \to
	\sum_{\vx \in \tilde{\Omega}}
	q^{ \lambda} (\vx^\prime  , \vx)
	\Go_\lambda (\vx) .
\end{align*}
Hence we may define $\widetilde{\Go} \in \joint (\{\Ao_\ell\}_{\ell =1}^m)$
by
\[
\widetilde{\Go} (\vx^\prime) 
:= \lim_{i \in I }  \Go_{\lambda^\prime (i)}  (\vx^\prime)
= \sum_{\vx \in \tilde{\Omega}} q^{ \lambda} (\vx^\prime  , \vx) \Go_\lambda (\vx) 
\]
By the last equality $[\widetilde{\Go}]$ is a lower bound of 
$([\Go_\lambda])_{\lambda \in \Lambda} .$

Now we take an observable $(\Bo (y))_{y \in \Omega^\prime}$
satisfying
$\Bo \ppeq \Go_{\lambda}$
for all $\lambda \in \Lambda .$
Then for each $\lambda \in \Lambda$ 
there exists a Markov kernel 
$r^\lambda \in \markov{\Omega^\prime}{\tilde{\Omega}}$
such that $\Bo = r^\lambda \ast \Go_\lambda .$
Since $\markov{\Omega^\prime}{\tilde{\Omega}}$ is compact, 
there exists a subnet $(r^{\lambda^{\prime\prime} (j)})_{j \in J}$
of $(r^{\lambda^{\prime} (i)})_{i \in I}$ such that
$r^{\lambda^{\prime\prime} (j)}$ converges to some
$\tilde{r} \in \markov{\Omega^\prime}{\tilde{\Omega}} .$
Then we have
\begin{align*}
	&\left\| \sum_{\vx \in \tilde{\Omega}}
	\tilde{r}(y , \vx) \widetilde{\Go} (\vx)
	-\Bo (y) \right\|
	\\
	&= \left\|
	\sum_{\vx \in \tilde{\Omega}}
	\tilde{r}(y , \vx) \widetilde{\Go} (\vx)
	-
	\sum_{\vx \in \tilde{\Omega}}
	r^{\lambda^{\prime\prime} (j)}(y , \vx) \Go_{\lambda^{\prime\prime} (j)} (\vx)
	\right\|
	\\
	&\leq
	\sum_{\vx \in \tilde{\Omega}}
	\tilde{r}(y , \vx) 
	\left\|
	\widetilde{\Go} (\vx) - \Go_{\lambda^{\prime\prime} (j)} (\vx) 
	\right\|
	+
	\sum_{\vx \in \tilde{\Omega}}
	| \tilde{r}(y , \vx) - r^{\lambda^{\prime\prime} (j)} (y , \vx) |
	\left\| \Go_{\lambda^{\prime\prime} (j)} (\vx) \right\|
	\\
	&\leq
	\sum_{\vx \in \tilde{\Omega}}
	\tilde{r}(y , \vx) 
	\left\|
	\widetilde{\Go} (\vx) - \Go_{\lambda^{\prime\prime} (j)} (\vx) 
	\right\|
	+
	\sum_{\vx \in \tilde{\Omega}}
	| \tilde{r}(y , \vx) - r^{\lambda^{\prime\prime} (j)} (y , \vx) |
	\\
	&\to 0 .
\end{align*}
This implies $\Bo = \tilde{r}\ast \widetilde{\Go} \ppeq \widetilde{\Go} ,$
which completes the proof. \qedhere
\end{enumerate}
\end{proof}

We can easily see that
$[\widetilde{\Go}]$ in Theorem~\ref{theo:min} is a supremum or an infinimum 
of $([\Go_\lambda])_{\lambda \in \Lambda} $
in $\joint (\{ \Ao_\ell \}_{\ell=1}^n) / \ppsim .$
Thus we obtain
\begin{corollary}
\label{coro:complete}
Let $\{\Ao_\ell\}_{\ell =1}^n$ be a finite set of compatible observables.
\begin{enumerate}[1.]
\item
The poset $\joint (\{\Ao_\ell\}_{\ell =1}^n) / \ppsim $
is both upper and lower directed complete (see Appendix \ref{sec:order}).
\item
Any joint observable $\Go \in \joint (\{\Ao_\ell\}_{\ell =1}^n) $ 
is upper bounded by a maximal joint observable and
lower bounded by a minimal joint observable.
\end{enumerate}
\end{corollary}

Unlike the minimal joint observables,
the set of maximal joint observables does depend on the choice of 
the post-processing representatives of the marginal observables 
as the following example demonstrates.

\begin{example}\label{ex:trivial}
Let $\Ao_1 = \Ao_2 = (\id_{\hi})$ be the single-outcome trivial observable.
Then $\Ao_1$ and $\Ao_2$ are compatible and the set 
$\joint (\{ \Ao_1 , \Ao_2 \})$
consists of only one element, 
the trivial observable $(\id_{\hi}) ,$
which is both minimal and maximal in $\joint (\{ \Ao_1 , \Ao_2 \}) .$

Clearly, the observables $\Ao_1 = \Ao_2$ are post-processing equivalent to the 
two-outcome trivial observable
$\Ao_1^\prime = \Ao_2^\prime = (\half \id_\hi , \half \id_\hi).$
If $\dim \hi \geq 2 ,$ there exists a non-trivial joint observable 
$\Go \in \joint  (\{ \Ao_1^\prime , \Ao_2^\prime \}) .$
Namely, fix an effect $0\neq T \neq \id$ and define
\begin{align*}
\Go(1,1)=\Go(2,2)=\half T \, , \quad \Go(1,2)=\Go(2,1)=\half (\id - T) \, .
\end{align*}
This is a joint observable of $\Ao_1^\prime$ and $\Ao_2^\prime$.
From Corollary~\ref{coro:complete} we conclude that $\Go$ is upper bounded by 
a maximal joint observable $\Go^\prime \in \joint  (\{ \Ao_1^\prime , \Ao_2^\prime \}),$
which is also non-trivial.
Hence there is no one-to-one correspondence between 
$\maxjoint (\{ \Ao_1 , \Ao_2 \})$ and
$\maxjoint  (\{ \Ao_1^\prime , \Ao_2^\prime \}) $
as in Proposition~\ref{prop:min_independence}.
\end{example}

\section{Characterization of minimal joint observable}
\label{sec:characterization}
In this section we give an algorithm to determine
whether a given joint observable is minimal or not.
Throughout this section, we fix a finite set of compatible observables
$\{ \Ao_\ell \}_{\ell =1}^n $ and
a joint observable
$\Go \in \joint (\{ \Ao_\ell \}_{\ell =1}^n) .$
Let $\Omega_\ell$ be the outcome set of $\Ao_\ell$
and denote $\tilde{\Omega} := \Omega_1 \times \dots \times \Omega_n .$

We need a bunch of auxiliary definitions before we can state our results.
For any two observables $\Ao$ and $\Bo$, with outcome sets $\Omega$ and $\Omega^\prime$, respectively,
we define
\[
	K(\Ao , \Bo)
	:=
	\left\{
	p \in \markov{\Omega}{\Omega^\prime}
	\mid
	\Ao = p \ast \Bo 
	\right\} ,
\]
which is, if nonempty, a compact convex subset of 
$\markov{\Omega}{\Omega^\prime} .$
Clearly, $K(\Ao,\Bo)$ is nonempty if and only if $\Ao \ppeq \Bo$.

We define
\[
	K_\Go
	:=
	\{
	p \in \markov{\tilde{\Omega}}{\tilde{\Omega}}
	\mid
	p \ast \Go \in \joint (\{ \Ao_\ell \}_{\ell =1}^n)
	\} .
\]
This is the set of those post-processings that are allowed for $\Go$ so that it is still a joint observable.
For $p \in \real^{\tilde{\Omega} \times \tilde{\Omega} }$, we have 
$p \in K_\Go$ if and only if $p$ satisfies the following linear (in)equalities:
\begin{gather*}
	\sum_{\vx^\prime \in \tilde{\Omega}}
	p(\vx^\prime , \vx)
	=1 
	\quad
	(\vx \in \tilde{\Omega}) ,
	\\
	\sum_{\vx^\prime \in \tilde{\Omega} , \pi_\ell (\vx^\prime) = x_\ell^\prime}
	\sum_{\vx \in \tilde{\Omega}}
	p(\vx^\prime , \vx) \Go (\vx)
	=
	\Ao_\ell (x_\ell^\prime)
	\quad
	(\ell \in \{  1, \dots , n\} ; x_\ell^\prime \in \Omega_\ell) ,
	\\
	p(\vx^\prime, \vx) \geq 0
	\quad
	(\vx , \vx^\prime \in \tilde{\Omega}) ,
\end{gather*}
where $\pi_\ell \colon \tilde{\Omega}\to \Omega_\ell$
is the canonical projection.
Thus, $K_\Go$ is a compact convex polytope on the Euclidean space
$\real^{\tilde{\Omega} \times \tilde{\Omega}} .$
It follows that the set of extreme points of $K_\Go$ is finite 
and can be written as $\{  p_1 , \dots , p_N\}$ (see Appendix~\ref{sec:polytope}).
We define 
\[
	p_\ast := \frac{1}{N} \sum_{k=1}^N p_k \in K_\Go \, .
\]

For each $\ell \in \{ 1 , \dots , n\} ,$
$K (\Ao_\ell , \Go)$ is also a nonempty compact convex polytope
on $\real^{\Omega_\ell \times \tilde{\Omega}} $
and the set of extreme points of $K (\Ao_\ell , \Go)$
can be written as
$\{q_{\ell , 1} , \dots , q_{\ell , N_\ell}  \} .$
We define $q_\ell \in  K (\Ao_\ell , \Go)$ by
\[
	q_\ell := \frac{1}{N_\ell} \sum_{k=1}^{N_\ell} q_{\ell , k}
\]
and $q_\ast \in K_\Go $ by
\[
	q_\ast (\vx^\prime , \vx)
	:=
	\prod_{\ell =1}^n q_{\ell} (\pi_\ell (\vx^\prime) , \vx) .
\]

Finally, for $r \in \markov{\tilde{\Omega}}{\tilde{\Omega}} ,$
$r$ is said to be \emph{conditionally independent}
if there exist Markov kernels 
$r_\ell \in \markov{\Omega_\ell}{\tilde{\Omega}}$
$(\ell \in \{ 1, \dots , n\})$
such that
\begin{equation}
	r (\vx^\prime , \vx )
	=
	\prod_{\ell =1}^n 
	r_\ell (\pi_\ell (\vx^\prime) , \vx) .
	\label{eq:rprod}
\end{equation}
We define
\[
K_\Go^{\mathrm{ind}} := \{ r \in K_\Go \mid \text{$r$ is conditionally independent}   \} .
\]

Now we can state and prove the following result.
\begin{theorem}
\label{theo:min1}
The following conditions are equivalent.
\begin{enumerate}[(i)]
\item \label{1:1} 
$\Go$ is a minimal joint observable.
\item \label{1:2}
For each $\vx_1 , \vx_2 , \vx^\prime \in \tilde{\Omega}$ and each 
$p \in K_\Go ,$
if $\Go (\vx_1)$ and $\Go (\vx_2)$ are linearly independent
then $p (\vx^\prime , \vx_1) p (\vx^\prime , \vx_2) = 0.$
\item \label{1:3}
For each $\vx_1 , \vx_2 , \vx^\prime \in \tilde{\Omega} ,$
if $\Go (\vx_1)$ and $\Go (\vx_2)$ are linearly independent
then $p_\ast (\vx^\prime , \vx_1) p_\ast (\vx^\prime , \vx_2) = 0.$
\item \label{1:4}
$p_\ast \ast \Go \ppsim \Go .$
\item \label{1:5}
$r \ast \Go \ppsim \Go$ for all $r \in K_\Go^{\mathrm{ind}} .$
\item \label{1:6}
For each $\vx_1 , \vx_2 , \vx^\prime \in \tilde{\Omega} $
and each $(r_\ell)_{\ell =1}^n \in K (\Ao_\ell , \Go) ,$
if $\Go (\vx_1)$ and $\Go (\vx_2)$ are linearly independent
then there exists $\ell \in \{ 1, \dots , n\}$ 
such that 
$r_\ell (\pi_\ell (\vx^\prime) , \vx_1) r_\ell (\pi_\ell (\vx^\prime) , \vx_2) = 0 .$
\item \label{1:6.2}
For each $\vx_1 , \vx_2 , \vx^\prime \in \tilde{\Omega} ,$
if $\Go (\vx_1)$ and $\Go (\vx_2)$ are linearly independent
then there exists $\ell \in \{ 1, \dots , n\}$ 
such that 
$p(\pi_\ell (\vx^\prime) , \vx_1) p (\pi_\ell (\vx^\prime) , \vx_2) = 0$
for all $p \in K (\Ao_\ell , \Go) .$ 
\item \label{1:7}
For each $\vx_1 , \vx_2 , \vx^\prime \in \tilde{\Omega} ,$
if $\Go (\vx_1)$ and $\Go (\vx_2)$ are linearly independent
then there exists $\ell \in \{ 1, \dots , n\}$ 
such that 
$q_\ell (\pi_\ell (\vx^\prime) , \vx_1) q_\ell(\pi_\ell (\vx^\prime) , \vx_2) = 0 .$
\item \label{1:8}
For each $\vx_1 , \vx_2 , \vx^\prime \in \tilde{\Omega} ,$
if $\Go (\vx_1)$ and $\Go (\vx_2)$ are linearly independent
then $q_\ast (\vx^\prime , \vx_1) q_\ast (\vx^\prime , \vx_2) = 0.$
\item \label{1:9}
$q_\ast \ast \Go \ppsim \Go .$
\end{enumerate}
\end{theorem}
\begin{proof}
The equivalences 
\eqref{1:1}$\iff$\eqref{1:2}, \eqref{1:3}$\iff$\eqref{1:4}, 
and \eqref{1:8}$\iff$\eqref{1:9} follow from Corollary~\ref{coro:mk}
and the definition of the minimality.
\eqref{1:1}$\implies$\eqref{1:5}, 
\eqref{1:2}$\implies$\eqref{1:3},
\eqref{1:5}$\implies$\eqref{1:9},
and
\eqref{1:6.2}$\implies$\eqref{1:6}
are obvious.
The equivalence \eqref{1:7}$\iff$\eqref{1:8} is immediate from the definition of $q_\ast .$
The equivalence \eqref{1:5}$\iff$\eqref{1:6} follows from Corollary~\ref{coro:mk}
and that $r$ given by \eqref{eq:rprod} is in $K_\Go^{\mathrm{ind}}$
if and only if $r_\ell \in K (\Ao_\ell , \Go)$
for each $\ell \in \{ 1 , \dots , n \} .$ 

\eqref{1:3}$\implies$\eqref{1:2}. 
Assume \eqref{1:3} and take an arbitrary element
$p \in K_\Go .$
By the finite-dimensional Krein-Milman theorem,
$p$ is a convex combination of $\{  p_1 , \dots , p_N \} ,$
that is, there exists $(\mu_k)_{k =1}^N \in [0,1]^N$ such that
$\sum_{k=1}^N \mu_k =1$ and
$p = \sum_{k=1}^N \mu_k p_k .$
Then
\[
	p(\vx^\prime , \vx)
	=
	\sum_{k=1}^N \mu_k p_k (\vx^\prime , \vx)
	\leq
	\sum_{k=1}^N  p_k (\vx^\prime , \vx)
	=
	N p_\ast (\vx^\prime , \vx).
\]
From this inequality and the assumption~\eqref{1:3},
the condition~\eqref{1:2} follows.

\eqref{1:7}$\implies$\eqref{1:6.2} can be shown similarly.

\eqref{1:6}$\implies$\eqref{1:2}.
Assume that (ii) is not true.
Then there exist $r \in K_\Go $ 
and
$\vx_0 , \vx_1 , \vx_2 \in \tilde{\Omega}  $
such that
$\Go (\vx_1)$ and $\Go (\vx_2)$ are linearly independent 
and
$r(\vx_0 , \vx_1) r(\vx_0 , \vx_2) \neq 0 .$
We define $r_\ell \in K (\Ao_\ell , \Go)$
by
\[
r_\ell (x_\ell^\prime , \vx) 
:= 
\sum_{\vx^\prime \in \tilde{\Omega}, \pi_\ell (\vx^\prime) = x_\ell^\prime} 
r(\vx^\prime , \vx).
\]
Then we have $r(\vx^\prime , \vx) \leq r_\ell (\pi_\ell (\vx^\prime) , \vx)$
and hence
\[
	0 <
	(r(\vx_0, \vx_1)   r(\vx_0, \vx_2)   )^{n}
	\leq
	\prod_{\ell =1}^n
	r_\ell (\pi_\ell (\vx_0), \vx_1) r_\ell (\pi_\ell (\vx_0), \vx_2) .
\]
Therefore we obtain
$r_\ell (\pi_\ell (\vx_0), \vx_1) r_\ell (\pi_\ell (\vx_0) , \vx_2) \neq 0$
for all $\ell \in \{ 1 , \dots , n\} ,$
proving that \eqref{1:6} is not true.
\end{proof}

Theorem~\ref{theo:min1} gives the following two algorithms to 
determine whether a given joint observable $\Go$ is minimal or not.
The first one is based on the condition~(iii).
According to Appendix~\ref{sec:polytope}, 
the set extreme points $\{ p_1 , \dots , p_N\} $
of $K_\Go$ can be
explicitly calculated, so can be $p_\ast  .$
The other one is based on the condition~\eqref{1:7} or \eqref{1:8},
which can be explicitly checked by calculating
the set of extreme points 
$\{ q_{\ell , 1} , \dots , q_{\ell , N_\ell} \}$
of the polytope $K (\Ao_\ell , \Fo)$
for each $\ell .$

At this point, we recall that an observable $\Ao$ with an outcome set $\Omega$ is said to be \emph{pairwise linearly independent}
if any pair $(\Ao (x_1) , \Ao(x_2))$, $x_1 , x_2 \in \Omega ; \, x_1 \neq x_2$,
is linearly independent.
Every observable is post-processing equivalent to a pairwise linearly independent observable unique up to the permutation of the outcome set \cite{MaMu90a,Kuramochi15b}.
We refer to Appendix \ref{sec:pp2} for further details on this property. 
In the following we develop a method to determine the minimality of a joint observable which is pairwise linearly independent.

Firstly, let $\Go$ be a joint observable.
For $ u \in \real^{\Omega_\ell \times \tilde{\Omega}} ,$
let us consider the following system of homogeneous linear (in)equalities:
\begin{gather}
	\sum_{\vx \in \tilde{\Omega}}
	u(x_\ell^\prime , \vx)
	\Go (\vx)
	=0 
	\quad
	(x_\ell^\prime \in \Omega_\ell) ,
	\label{eq:C1}
	\\
	\sum_{x_\ell^\prime \in \Omega_\ell}
	u(x_\ell^\prime , \vx)
	=0
	\quad
	(\vx \in \tilde{\Omega}),
	\label{eq:C2}
	\\
	u (x_\ell^\prime , \vx)
	\begin{cases}
		\leq 0  & \text{if $x_\ell^\prime = \pi_\ell (\vx) ,$} \\
		\geq 0  & \text{if $x_\ell^\prime \neq \pi_\ell (\vx) .$} 
	\end{cases}
	\label{eq:C3}
\end{gather}
We denote by $C_\ell (\Go)$ the pointed polyhedral cone 
on $\real^{\Omega_\ell \times \tilde{\Omega}}$
consisting of the solutions of
\eqref{eq:C1}, \eqref{eq:C2}, and \eqref{eq:C3}.

\begin{proposition}
\label{prop:min2}
Suppose that $\Go$ is pairwise linearly independent.
Then $\Go$ is a minimal joint observable if and only if
$C_\ell (\Go) = \{ 0 \}$ for all $\ell \in \{ 1 ,\dots , n\} .$
\end{proposition}
\begin{proof}
Assume that $C_\ell (\Go) \neq \{ 0 \}$ for some $\ell .$
Then we can take $ u \in C_\ell (\Go)  ,$
$x_\ell^0 \in \Omega_\ell ,$
and $\vx_1 \in \tilde{\Omega}$
such that 
$u (x_\ell^0 , \vx_1) \neq 0 .$
If $x_\ell^0 = \pi_\ell (\vx_1) ,$
then by \eqref{eq:C2}
there exists another $x_\ell^{1} \in \Omega_\ell$
satisfying $x_\ell^{1} \neq x_\ell^0$
and $p (x_\ell^{1} , \vx_1) \neq 0 .$
Thus, by also considering~\eqref{eq:C3},
we may assume $x_\ell^0 \neq \pi_\ell (\vx_1)$
and $ u (x_\ell^0 , \vx_1) > 0 .$
We define $\vx_2 \in \tilde{\Omega}$
by
\[
	\pi_{\ell^\prime} (\vx_2)
	=
	\begin{cases}
		x_\ell^0 & \text{if $\ell^\prime = \ell ,$} \\
		\pi_{\ell^\prime} (\vx_1) & \text{if $\ell^\prime \neq \ell . $}
	\end{cases}
\]
By multiplying a small positive constant if necessary, 
we may assume 
\begin{equation}
	|u (x_\ell^\prime , \vx)| <1
	\quad
	\forall (x_\ell^\prime , \vx) \in \Omega_\ell \times \tilde{\Omega} .
	\label{eq:C4}
\end{equation}
We define $r_\ell \in \real^{\Omega_\ell  \times \tilde{\Omega}}$
by
\[
	r_\ell (x_\ell^\prime , \vx)
	:=
	\delta_{x_\ell^\prime , \pi_\ell ( \vx )} 
	+ 
	u (x_\ell^\prime , \vx) .
\]
Then from \eqref{eq:C1}, \eqref{eq:C2}, \eqref{eq:C3},
and \eqref{eq:C4}, 
we have 
$r_\ell \in K (\Ao_\ell, \Go) .$
Furthermore,
we have $\vx_1 \neq \vx_2$
and 
\begin{gather*}
	r_\ell (x_\ell^0 , \vx_1)
	=
	u (x_\ell^0 , \vx_1) \neq 0 ,
	\\
	r_\ell (x_\ell^0 , \vx_2)
	=
	1 + u  (x_\ell^0 , \vx_2) \neq 0 .
\end{gather*}
We define $r \in \markov{\tilde{\Omega}}{\tilde{\Omega}}$
by
\begin{gather*}
	r ( \vx^\prime , \vx   )
	:=
	r_\ell (\pi_\ell (\vx^\prime) , \vx)
	\prod_{\ell^\prime \neq \ell}
	\delta_{ \pi_{\ell^\prime} (\vx^\prime) , \pi_{\ell^\prime} (\vx)  } .
\end{gather*}
Then $r \ast \Go \in \joint (\{\Ao_{\ell^\prime}  \}_{\ell^\prime =1}^n)$
and 
\begin{gather*}
	r ( \vx_2 , \vx_1)
	=
	r_\ell (x_\ell^0 , \vx_1) \neq 0 ,
	\\
	r ( \vx_2 , \vx_2)
	=
	r_\ell (x_\ell^0 , \vx_2)
	\neq 0 .
\end{gather*}
This implies that the condition~\eqref{1:2} of Theorem~\ref{theo:min1} is not true.
Therefore $\Go$ is not minimal.

Assume that $C_\ell (\Go) = \{ 0\}$ for all
$\ell \in \{ 1 ,\dots , n \} .$
Let $\Go^\prime \in \joint (\{\Ao_{\ell^\prime}  \}_{\ell^\prime =1}^n)$
be a  joint observable satisfying $\Go^\prime \ppeq \Go .$
Then we can take $r  \in \markov{\tilde{\Omega}}{\tilde{\Omega}}$
such that $\Go^\prime = r \ast \Go .$
For each $\ell$ we define
$r_\ell \in K (\Ao_\ell , \Go) $ by
\[
	r_\ell (x_\ell^\prime , \vx)
	:=
	\sum_{\vx^\prime , \pi_\ell (\vx^\prime) = x_\ell^\prime}
	r (\vx^\prime , \vx) ,
\]
and 
$u_\ell \in \real^{\Omega_\ell \times \tilde{\Omega}}$
by
\[
	u_\ell (x_\ell^\prime , \vx)
	:=
	r_\ell (x_\ell^\prime , \vx) 
	- \delta_{ x_\ell^\prime , \pi_\ell (\vx )  } .
\]
From $r_\ell(x_\ell^\prime , \vx)  \in [0,1]$
we can easily check that $u_\ell \in C_\ell (\Go ) .$
Therefore by assumption we have $u_\ell = 0$
and hence
$r_\ell  (x_\ell^\prime , \vx) = \delta_{ x_\ell^\prime , \pi_\ell (\vx )  } . $
If $\vx^\prime \neq \vx ,$ there exists $\ell$ with 
$\pi_\ell (\vx^\prime) \neq \pi_\ell (\vx) .$
Thus
\[
	r(\vx^\prime , \vx)
	\leq
	r_\ell (\pi_\ell (\vx^\prime) , \vx)
	=
	\delta_{\pi_\ell (\vx^\prime) , \pi_\ell (\vx)}
	=0.
\]
This implies $r (\vx^\prime , \vx) = \delta_{\vx , \vx^\prime}$
and hence 
$\Go^\prime = r \ast  \Go = \Go .$
Thus $\Go$ is minimal.
\end{proof}

For each $\ell ,$
$C_\ell (\Go)$ is a polyhedral cone and is the conical hull of a finite set
$\{ u_{\ell 1} , \dots , u_{\ell , M_\ell } \} ,$
which can be calculated according to Appendix~\ref{sec:polytope}.
Thus Proposition~\ref{prop:min2} provides a method to 
determine the minimality of $\Go$ when $\Go$ is pairwise linearly independent.

\begin{proposition}
\label{prop:independent}
Let 
$
\tilde{\Omega}_1 := 
\{ \vx \in \tilde{\Omega} \mid \Go (\vx) \neq 0 \}.
$
Suppose that $(\Go (\vx) )_{\vx \in \tilde{\Omega}_1}$ is linearly independent.
Then $\Go$ is a minimal joint observable.
\end{proposition}

\begin{proof}
Let $\Go^\prime \in \joint (\{ \Ao_\ell \}_{\ell =1}^n)$ be a joint observable satisfying
$\Go^\prime \ppeq \Go$
and let $p \in \markov{\tilde{\Omega}}{\tilde{\Omega}}$ be a Markov kernel 
such that
$
\Go^\prime = p \ast \Go .
$
From $\Go ,\Go^\prime \in \joint (\{ \Ao_\ell \}_{\ell =1}^n )$
we have
\begin{align*}
	\sum_{\vx \in \tilde{\Omega}_1 ,  \pi_\ell (\vx) = x_\ell  }
	\Go (\vx)
	&=
	\Ao_\ell (x_\ell)
	= 
	\sum_{\vx^\prime \in \tilde{\Omega} , \pi_\ell (\vx^\prime) = x_\ell}
	\Go^\prime (\vx^\prime)
	\\ 
	&=
	\sum_{\vx \in \tilde{\Omega}_1} 
	\left(
	\sum_{\vx^\prime \in \tilde{\Omega} , \pi_\ell (\vx^\prime) = x_\ell }
	p(\vx^\prime , \vx) 
	\right)
	\Go (\vx) .
\end{align*}
Thus the linear independence of $(\Go (\vx) )_{\vx \in \tilde{\Omega}_1}$ implies
\[
	\sum_{\vx^\prime , \pi_\ell( \vx^\prime ) = x_\ell}  p(\vx^\prime , \vx)
	=
	\delta_{\pi_\ell(\vx) , x_\ell} 
\]
for all $\vx \in \tilde{\Omega}_1 .$
By the same discussion as in Proposition~\ref{prop:min2},
we obtain $p(\vx^\prime  ,\vx) = \delta_{\vx , \vx^\prime}$
($\forall \vx \in \tilde{\Omega}_1 ,$  $\forall \vx^\prime \in \tilde{\Omega}$)
and hence $\Go = \Go^\prime .$
Therefore, $\Go$ is minimal.
\end{proof}

\begin{corollary}
\label{coro:independent}
Suppose that $\Go$ is linearly independent.
Then $\Go$ is both minimal and maximal.
\end{corollary}

\begin{proof}
The minimality of $\Go$ is immediate from Proposition~\ref{prop:independent}.

Assume $\Go \ppeq \Go^\prime $
for some $\Go^\prime \in \joint (\{ \Ao_\ell \}_{\ell =1}^n) .$
Since each element of $\Go$ is a linear combination of $\Go^\prime ,$
a dimensional argument yields that 
$\Go^\prime$ is also linearly independent.
(Here we used the finiteness of the outcome set 
$\tilde{\Omega}$).
Hence from the proof of Proposition~\ref{prop:independent} we obtain 
$\Go = \Go^\prime .$
Therefore, $\Go$ is maximal.
\end{proof}

\section{Dichotomic qubit observables} \label{sec:dichotomic}

\subsection{Compatibility of dichotomic qubit observables}

We denote  $\vsigma =(\sigma_1,\sigma_2,\sigma_3)$.
A dichotomic qubit observable is of the form
$$
\Eaa(\pm) = \half ( \alpha \id \pm \va \cdot \vsigma ) \, , 
$$
where $\va\in\real^3$, $\alpha\in[0,2]$ and $\no{\va} \leq \min(\alpha,2-\alpha)$.

All joint observables of a compatible pair $(\Eaa,\Ebb)$ are parametrized by two parameters $\gamma\in\real$ and $\vg\in\real^3$ in the following way. 
The joint observable $\Gg$ is defined as
\begin{align*}
\Gg(+,+) &= \half ( \gamma \id + \vg \cdot \vsigma ) ,\\
\Gg(+,-) &= \Ea(+) - \Gg(+,+) \\
&= \half [ (\alpha - \gamma)\id +(\va- \vg ) \cdot \vsigma ] ,\\
\Gg(-,+) &= \Eb(+) - \Gg(+,+)  \\
& = \half [ (\beta - \gamma)\id +(\vb- \vg ) \cdot \vsigma ] \\
\Gg(-,-) &= \id + \Gg(+,+) - \Ea(+) - \Eb(+) 
\\
&=  \half [ (2 + \gamma - \alpha - \beta) \id +(\vg- \va -\vb ) \cdot \vsigma ].
\end{align*}
For $\Gg$ to be a valid joint observable, these four operators must be positive. 
This means that the parameters $\gamma,\vg$ have to satisfy the following inequalities \cite{StReHe08}:
\begin{align}
& \no{\vg} \leq \gamma \label{eq:g1}\\ 
& \no{\va - \vg} \leq \alpha-\gamma  \label{eq:g2}  \\
& \no{\vb - \vg} \leq \beta- \gamma \label{eq:g3} \\
& \no{\va+\vb-\vg} \leq 2 + \gamma - \alpha - \beta \label{eq:g4}
\end{align} 

From the previous inequalities one can solve the compatibility condition for $\Eaa$ and $\Ebb$, and three equivalent formulations are presented in \cite{StReHe08,BuSc10,YuLiLiOh10}.
For our purposes, we do not need the compatibility condition; we simply assume that $\Eaa$ and $\Ebb$ are compatible and in the following we analyze the condition for the minimality of joint observable.

\subsection{Minimal joint observables}

Now we give a complete characterization of the minimality of 
$\Gg$
when the marginal observables $\Eaa$ and $\Ebb$ are non-commutative,
which holds if and only if $\va$ and $\vb$ are linearly independent.
We put $\Ao_1 := \Eaa $ and $\Ao_2 := \Ebb $ and adopt the same notation as 
in Section~\ref{sec:characterization}.

\subsubsection{Trivial compatibility}

Two dichotomic observables $\Ao_1$ and $\Ao_2$ are trivially compatible if one of the orderings  $\Ao_1(+) \leq \Ao_2(+)$,  $\Ao_1(+) \geq \Ao_2(-)$,  $\Ao_1(+) \leq \Ao_2(-)$ or  $\Ao_1(+) \geq \Ao_2(-)$ holds \cite{BuHe08}.
Two compatible dichotomic observables satisfy this kind of trivial condition exactly when they have a joint observable with one element being zero.

Suppose that $\Gg(++)=0$, which is equivalent to $\Ao_1(+) \leq \Ao_2(-)$. 
In this case, the other elements of $\Gg$ are determined to be
\begin{align*}
\Gg(+ -) &= \Ao_1(+) \, , \\
\Gg(- +) & = \Ao_2(+) \, , \\
\Gg(- -)
&=  \id - \Ao_1(+)-\Ao_2(+) \, .
\end{align*}
Analogous equations follow in other cases when $\Gg(+-)=0$, $\Gg(-+)=0$ or $\Gg(--)=0$. 

We conclude that in the case of trivial compatibility, a joint measurement has three nonzero elements and it is unique.

\subsubsection{Linearly independent vectors}

Secondly, we consider the case when the vectors $\va , \vb ,\vg$ are linearly independent.
In this case, we can easily check that 
$\Gg$ is linearly independent.
Hence Corollary~\ref{coro:independent} implies that 
$\Gg$ is both maximal and minimal.

\subsubsection{Linearly dependent vectors}

Thirdly, we consider the case when $\vg$ can be written as a linear combination $\vg = c_1 \va + c_2 \vb$, $c_1 , c_2 \in \real$.
In order to calculate $K(\Ao_\ell , \Gg)$
and $C_\ell (\Gg)$, consider the following homogeneous linear (in)equalities for $u \in \real^{\tilde{\Omega}} :$
\begin{gather}
	\begin{pmatrix}
		\gamma & \alpha - \gamma & \beta - \gamma & 2 + \gamma - \alpha - \beta \\
		c_1 & 1-c_1 & -c_1 & c_1 -1  \\
		c_2 & -c_2 & 1-c_2 & c_2-1
	\end{pmatrix}
	\begin{pmatrix}
		u (++) \\ u (+-) \\ u (-+) \\ u(--)
	\end{pmatrix}
	=0,
	\label{eq:ueq}
	\\
	u (x , x^\prime)
	\begin{cases}
		\leq 0 & \text{if $x = + $} \\
		\geq 0 & \text{if $x = -  ,$}
	\end{cases}
	\label{eq:uineq1}
	\\
	u (x , x^\prime)
	\begin{cases}
		\leq 0 & \text{if $x^\prime = + $} \\
		\geq 0 & \text{if $x^\prime = -  .$}
	\end{cases}
	\label{eq:uineq2}
\end{gather}
Let $C_{1,+} $ be the polyhedral cone on $\real^{\tilde{\Omega}}$ 
defined by \eqref{eq:ueq} and \eqref{eq:uineq1},
and $C_{2,+}$ be the one defined by
\eqref{eq:ueq} and \eqref{eq:uineq2}.
Then it can be checked that $K(\Ao_\ell , \Gg)$
and $C_\ell (\Gg)$
can be given as follows.
\begin{itemize}
\item
For $p \in \real^{\Omega_\ell \times \tilde{\Omega}} ,$
$p \in K ( \Ao_\ell , \Gg  )$
if and only if
there exists  $u_+ \in C_{\ell , +}$ 
such that
$p (+ , \vx) = \delta_{+ , \pi_\ell (\vx)} + u_+ (\vx)$
and
$p (- , \vx) = \delta_{- , \pi_\ell (\vx)} - u_+ (\vx) \geq 0 $
$(\vx \in \tilde{\Omega}) .$
\item
For $v \in \real^{\Omega_\ell \times \tilde{\Omega}} , $
$v \in C_\ell (\Gg)$
if and only if there exists $u_+ \in C_{\ell , +}$
such that
$v (+ , \vx) = u_+ (\vx)$
and
$v (- , \vx) = - u_+ (\vx) $
$(\vx \in \tilde{\Omega}) .$
\end{itemize}
The general solution of \eqref{eq:ueq} is
$u = tw$
$(t \in \real),$
where $w \in \real^{\tilde{\Omega}}$ is given by
\begin{equation}\label{eq:w}
	\begin{pmatrix}
		w (++) \\ w(+-) \\ w(-+) \\ w(--)
	\end{pmatrix}
	:=
	\begin{pmatrix}
		(2-\alpha)c_1 + (2-\beta)c_2 + \gamma -2 \\
		(2-\alpha)c_1 -\beta c_2 + \gamma  \\
		-\alpha c_1 + (2-\beta)c_2 + \gamma \\
		- \alpha c_1  -\beta c_2 + \gamma
	\end{pmatrix} 
	.
\end{equation}
If $\Gg$ is pairwise linearly independent (see Appendix \ref{sec:pp2}),
then by Proposition~\ref{prop:min2}, 
$\Gg$ is a minimal joint observable if and only if
$C_1 (\Gg) = C_2 (\Gg) = \{ 0 \} ,$
or equivalently,
$C_{1 , + }  = C_{2 , +} = \{ 0\} .$
By noting $w \neq 0 ,$
$C_{1 , +} \neq \{ 0 \}$ if and only if
\begin{align*}
	&\left[
	w(++) \leq 0 
	\, \land \,
	w(+-) \leq 0
	\, \land \,
	w (-+) \geq 0
	\, \land \,
	w(--) \geq 0
	\right]
	\\
	\lor
	&\left[
	w(++) \geq 0 
	\, \land \,
	w(+-) \geq 0
	\, \land \,
	w (-+) \leq 0
	\, \land \,
	w(--) \leq 0 
	\right]
\end{align*}
Hence $C_{1,+} = \{ 0 \}$ if and only if
\begin{align*}
	&\left[
	w(++) > 0 
	\, \lor \,
	w(+-) > 0
	\, \lor \,
	w (-+) < 0
	\, \lor \,
	w(--) < 0
	\right]
	\\
	\land
	& \left[
	w(++) < 0 
	\, \lor \,
	w(+-) < 0
	\, \lor \,
	w (-+) > 0
	\, \lor \,
	w(--) > 0
	\right] .
\end{align*}
By using the equivalences
\begin{gather*}
	(x >0 \land y>0)  \, \lor \, (x < 0 \land y<0)
	\iff
	xy >0 ,
	\\
	(x >0 \land y < 0)  \, \lor \, (x < 0 \land y>0)
	\iff
	xy < 0 
\end{gather*}
valid for $x,y \in \real ,$
we obtain
\begin{align*}
	&C_{1,+} = \{ 0 \}
	\\
	\iff
	&
	(w (++) w(+-) <0)
	\, \lor \,
	(w (++) w(-+) > 0)
	\, \lor \,
	(w (++) w(--) >0)
	\\
	&\lor \,
	(w (+-) w(-+) >0)
	\, \lor \,
	(w (+-) w(--) >0)
	\, \lor \,
	(w (-+) w(--) <0)
\end{align*}
Similarly we obtain
\begin{align*}
	&C_{2,+} = \{ 0 \}
	\\
	\iff
	&
	(w (++) w(+-) >0)
	\, \lor \,
	(w (++) w(-+) < 0)
	\, \lor \,
	(w (++) w(--) >0)
	\\
	&\lor \,
	(w (+-) w(-+) >0)
	\, \lor \,
	(w (+-) w(--) <0)
	\, \lor \,
	(w (-+) w(--) >0) .
\end{align*}
Hence 
\begin{align}
	&C_{1,+} =  C_{2,+} = \{ 0 \} 
	\notag 
	\\ 
	\iff  &
	(w(++) w(--) > 0) 
	\, \lor \,
	(w(+-) w(-+) > 0)
	\notag 
	\\
	& \lor \, 
	\left\{
	\left[
	 (w(++) w(+-) <0)
	\, \lor \,
	(w(++) w(-+) >0)
	\, \lor \,
	(w(+-) w(--) >0)
	\right. \right.
	\notag 
	\\
	& \lor \,
	\left. 
	(w(-+) w(--) <0)
	\right]
	\, \land \,
	\left[
	(w(++) w(+-) >0)
	\, \lor \,
	(w(++) w(-+) <0)
	\right. 
	\notag 
	\\
	& 
	\left. \left.
	\lor \,
	(w(+-) w(--) <0)
	\, \lor \,
	(w(-+) w(--) >0)
	\right] \right\}
	\notag 
	\\
	\iff  &
	(w(++) w(--) > 0) 
	\, \lor \,
	(w(+-) w(-+) > 0)
	\notag 
	\\
	& \lor \,
	(w(++)^2 w(+-)  w(-+) >0 )
	\, \lor \,
	(w(++) w(+-)^2  w(--) >0 )
	\notag 
	\\
	& \lor \,
	( w (++)  w (+-) w (-+) w (--) <0 ) 
	\, \lor \,
	( w (++)   w (-+)^2  w (--)  >0 ) 
	\notag 
	\\
	& \lor \,
	(  w (+-) w (-+) w (--)^2 > 0 ) 
	\notag 
	\\
	\iff &
	(w(++) w(--) > 0) 
	\, \lor \,
	(w(+-) w(-+) > 0)
	\, \lor \,
	( w (++)  w (+-) w (-+) w (--) <0 ) 
	\notag 
	\\
	\iff&
	(w(++) w(--) > 0) 
	\, \lor \,
	(w(+-) w(-+) > 0)
	\label{eq:wmin}
\end{align}

Now we consider the case when $\Gg$ is not pairwise linearly independent.
The linear dependence conditions for the pairs of the elements of $\Gg$
are given as follows.
\begin{enumerate}[(i)]
\item
$\Gg (++)$ and $\Gg (+-)$
are linearly dependent
if and only if
\begin{equation}
	\vg = \frac{\gamma}{\alpha} \va .
	\label{eq:dep1}
\end{equation}
\item
$\Gg (++)$ and $\Gg (-+)$
are linearly dependent
if and only if
\begin{equation}
	\vg = \frac{\gamma}{\beta} \vb .
	\label{eq:dep2}
\end{equation}
\item
$\Gg (++)$ and $\Gg (--)$
are linearly dependent
if and only if
\begin{equation}
	\alpha + \beta \neq 2
	\, \land \,
	\vg = \frac{\gamma}{\alpha + \beta -2} (\va + \vb) .
	\label{eq:dep3}
\end{equation} 
\item
$\Gg (+-)$ and $\Gg (-+)$
are linearly dependent
if and only if
\begin{equation}
	\alpha \neq \beta
	\, \land \,
	\vg = 
	\frac{\gamma - \beta}{\alpha -\beta} \va 
	+ \frac{\gamma - \alpha}{\beta - \alpha} \vb .
	\label{eq:dep4}
\end{equation}
Here, the possibility of $\alpha = \beta $ will be excluded 
since this implies $\vg = \va = \vb ,$
contradicting the linear independence of $\va$ and $\vb .$
\item
$\Gg (+-)$ and $\Gg (--)$
are linearly dependent
if and only if
\begin{equation}
	\beta \neq 2
	\, \land \, 
	\vg = \va + \frac{\alpha - \gamma}{2- \beta} \vb .
	\label{eq:dep5}
\end{equation}
Here, the possibility of $\beta =2 $ will be excluded 
since this implies $  \va = \vg = \va + \vb ,$ 
contradicting $\vb \neq 0 .$
\item
$\Gg (-+)$ and $\Gg (--)$
are linearly dependent
if and only if
\begin{equation}
	\alpha \neq 2
	\, \land \,
	\vg = \frac{\beta - \gamma}{2- \alpha} \va + \vb .
	\label{eq:dep6}
\end{equation}
Here,
the possibility of $\alpha =2 $ will be excluded 
since this implies $  \vb = \vg = \va + \vb ,$ 
contradicting $\va \neq 0 .$
\end{enumerate}
As a conclusion, $\Gg$ is pairwise linearly independent if and only if neither of the above 
six conditions holds.
We remark that each of the above conditions corresponds to 
the intersection point of two lines 
$w (\vx) = 0$ and $w (\vx^\prime) =0 $
on $c_1$-$c_2$ plane.
For example, the condition~\eqref{eq:dep1} holds if and only if
$w (-+) = w(--) =0  .$

Now assume that $\Gg (++) $ and $\Gg (+-)$
are linearly dependent.
Then
$\vg = \frac{\gamma}{\alpha} \va ,$
and
$c_1 = \frac{\gamma}{\alpha} ,$ $c_2 =0 .$
Hence
\[
	w (++) = -2 \frac{  \alpha -  \gamma}{\alpha} , \,
	w (+-) = 2 \frac{\gamma }{\alpha} , \,
	w(-+) = w(--) = 0.
\]
Thus each $r_ \ell \in K (\Ao_\ell , \Gg)$ $(\ell =1,2)$
can be written as
\begin{gather*}
	\begin{pmatrix}
		r_1 (+, ++) \\ r_1 (+, +-) \\ r_1 (+, -+) \\ r_1 (+, --)
	\end{pmatrix}
	=
	\begin{pmatrix}
		1  \\ 1 \\ 0 \\ 0
	\end{pmatrix}
	, \quad
	\begin{pmatrix}
		r_1 (-, ++) \\ r_1 (-, +-) \\ r_1 (-, -+) \\ r_1 (-, --)
	\end{pmatrix}
	=
	\begin{pmatrix}
		0 \\ 0  \\ 1 \\ 1
	\end{pmatrix} 
	,
	\\
	\begin{pmatrix}
		r_2 (+, ++) \\ r_2 (+, +-) \\ r_2 (+, -+) \\ r_2 (+, --)
	\end{pmatrix}
	=
	\begin{pmatrix}
		1 - 2 s \frac{\alpha - \gamma }{\alpha}  \\ 2 s \frac{\gamma }{\alpha}  \\ 1 \\ 0
	\end{pmatrix}
	, \quad
	\begin{pmatrix}
		r_2 (-, ++) \\ r_2 (-, +-) \\ r_2 (-, -+) \\ r_2 (-, --)
	\end{pmatrix}
	=
	\begin{pmatrix}
		2 s \frac{\alpha - \gamma }{\alpha} \\ 1 -  2s  \frac{\gamma }{\alpha}  \\ 0 \\ 1
	\end{pmatrix}
\end{gather*}
for some $s \geq 0 .$
We can easily check the condition~\eqref{1:6} of Theorem~\ref{theo:min1},
which implies the minimality of $\Gg .$

We can similarly check the minimality of $\Gg$ 
when either $(\Gg (++) , \Gg (-+)) ,$ $(\Gg (+-) , \Gg (--)) ,$
or $(\Gg (-+) , \Gg (--))$ is a linearly dependent pair. 

Assume that $\Gg (++)$ and $\Gg (--)$ are linearly dependent.
Then 
$\alpha + \beta \neq 2 $
and $\vg = \frac{\gamma }{\alpha + \beta -2} (\va + \vb) .$
Thus 
\[
	w (++) = -2 \frac{2+ \gamma - \alpha - \beta}{2-\alpha - \beta} , \,
	w (-+) = w(+-) = 0 , \,
	w (--) = \frac{2\gamma}{2- \alpha - \beta } .
\]
Hence
each $r_\ell \in K (\Ao_\ell , \Gg)$
$(\ell =1,2)$
can be written as
\begin{gather*}
	\begin{pmatrix}
		r_1 (+, ++) \\ r_1 (+, +-) \\ r_1 (+, -+) \\ r_1 (+, --)
	\end{pmatrix}
	=
	\begin{pmatrix}
		1 - t (2+ \gamma - \alpha - \beta)  \\ 1 \\ 0 \\ t \gamma 
	\end{pmatrix}
	, \quad
	\begin{pmatrix}
		r_1 (-, ++) \\ r_1 (-, +-) \\ r_1 (-, -+) \\ r_1 (-, --)
	\end{pmatrix}
	=
	\begin{pmatrix}
		 t (2+ \gamma - \alpha - \beta) \\ 0 \\ 1 \\ 1 - t\gamma
	\end{pmatrix} 
	,
	\\
	\begin{pmatrix}
		r_2 (+, ++) \\ r_2 (+, +-) \\ r_2 (+, -+) \\ r_2 (+, --)
	\end{pmatrix}
	=
	\begin{pmatrix}
		1 - s(2+ \gamma - \alpha - \beta)  \\ 0 \\ 1 \\ s\gamma
	\end{pmatrix}
	, \quad
	\begin{pmatrix}
		r_2 (-, ++) \\ r_2 (-, +-) \\ r_2 (-, -+) \\ r_2 (-, --)
	\end{pmatrix}
	=
	\begin{pmatrix}
		 s (2+ \gamma - \alpha - \beta) \\ 1 \\ 0 \\ 1 -s\gamma
	\end{pmatrix}
\end{gather*}
for some $t, s \geq 0 .$
By taking sufficiently small $t , s > 0 ,$
we have
\begin{gather*}
	r_1 (+ , ++ ) r_1 (+ , +-) \neq 0 \neq r_2(- , ++) r_2 (-,+-) .  
\end{gather*}
Since $(\Gg (++) , \Gg (+-))$ is a linearly independent pair in this case,
Theorem~\ref{theo:min1} implies that $\Gg$ is not minimal.

Assume that $\Gg (+-)$ and $\Gg (-+)$ are linearly dependent.
Then $\alpha \neq \beta$ and 
$\vg = \frac{\gamma - \beta}{\alpha -\beta} \va + \frac{\gamma - \alpha}{\beta - \alpha} \vb .$
Thus
\[
	w(++) = w(--) =0 ,
	\,
	w(+-)
	=
	2 \frac{\beta - \gamma }{\beta - \alpha} ,
	\,
	w(-+)
	=
	2 \frac{\alpha - \gamma }{ \alpha - \beta} . 
\]
Hence
each $r_\ell \in K (\Ao_\ell , \Gg)$
$(\ell =1,2)$
can be written as
\begin{gather*}
	\begin{pmatrix}
		r_1 (+, ++) \\ r_1 (+, +-) \\ r_1 (+, -+) \\ r_1 (+, --)
	\end{pmatrix}
	=
	\begin{pmatrix}
		1  \\ 1 - t (\beta - \gamma ) \\  t (\alpha - \gamma ) \\ 0
	\end{pmatrix}
	, \quad
	\begin{pmatrix}
		r_1 (-, ++) \\ r_1 (-, +-) \\ r_1 (-, -+) \\ r_1 (-, --)
	\end{pmatrix}
	=
	\begin{pmatrix}
		 0 \\ t (\beta - \gamma) \\ 1 - t (\alpha - \gamma)  \\ 1 
	\end{pmatrix} 
	,
	\\
	\begin{pmatrix}
		r_2 (+, ++) \\ r_2 (+, +-) \\ r_2 (+, -+) \\ r_2 (+, --)
	\end{pmatrix}
	=
	\begin{pmatrix}
		1   \\ s (\beta - \gamma) \\ 1 - s( \alpha - \gamma) \\ 0
	\end{pmatrix}
	, \quad
	\begin{pmatrix}
		r_2 (-, ++) \\ r_2 (-, +-) \\ r_2 (-, -+) \\ r_2 (-, --)
	\end{pmatrix}
	=
	\begin{pmatrix}
		 0 \\ 1 - s (\beta - \gamma) \\ s (\alpha - \gamma) \\ 1  
	\end{pmatrix}
\end{gather*}
for some $t, s \geq 0 .$
By taking sufficiently small $t,s >0 ,$ we have
\[
	r_1 (+ , ++ ) r_1 (+ , +-) \neq 0 \neq r_2(+ , ++) r_2 (+,+-) .
\]
Since $\Gg (++)$ and $\Gg (+-)$ is linearly independent in this case,
Theorem~\ref{theo:min1} implies that $\Gg$ is not minimal.

We remark that the above two non-minimal cases \eqref{eq:dep3}
and \eqref{eq:dep4} are inconsistent with the condition~\eqref{eq:wmin}.

To summarize, we obtain the following proposition which 
completely characterizes the minimality $\Gg .$

\begin{theorem}
\label{theo:qubitmin}
Let $\Eaa$ and $\Ebb$ be compatible qubit observables.
Suppose that $\va$ and $\vb$ are linearly independent.
Then their joint observable $\Gg$ is a minimal joint observable if and only if 
either of the following conditions holds.
\begin{enumerate}[1.]
\item
$\vg,$ $\va ,$ and $\vb$ are linearly independent.
\item
One (and only one) element of $\Gg$ is zero.
\item
$\Gg$ satisfies either of the pairwise linear dependence conditions
\eqref{eq:dep1}, \eqref{eq:dep2}, \eqref{eq:dep5}, and \eqref{eq:dep6}.
\item
$\vg$ can be written as $\vg = c_1 \va + c_2 \vb $
$(c_1 , c_2 \in \real)$
and satisfies \eqref{eq:wmin}.
\end{enumerate}
\end{theorem}

\subsection{Unbiased dichotomic qubit observables}

Finally, we consider the simple special case $\alpha = \beta= 1$.
These kind of observables are called \emph{unbiased}.
As shown in \cite{Busch86}, $\Ea$ and $\Eb$ are compatible if and only if 
\begin{equation}
\no{\va - \vb} + \no{\va + \vb} \leq 2 \, .
\end{equation}
If $\va$ and $\vb$ are nonzero different vectors, then we easily see that two unbiased observables $\Ea$ and $\Eb$ cannot be trivially compatible. 
In particular, if $\va$ and $\vb$ are linearly independent, then all elements of $\Gg$ must be nonzero.

\begin{figure}
\centering
\includegraphics[width=7.5cm]{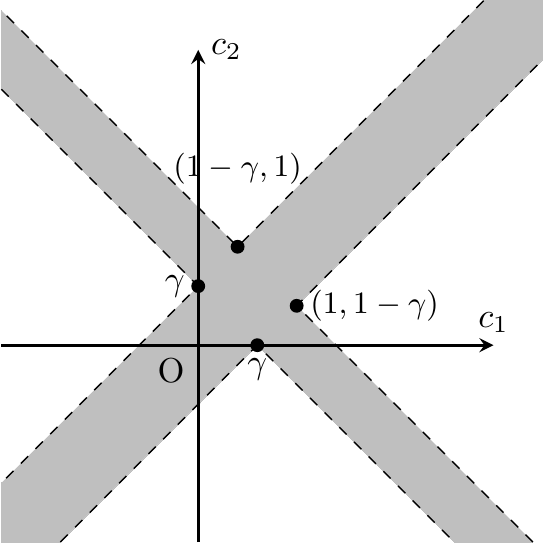}
\caption{The region of $(c_1 , c_2)$ such that $\Gg$ is a minimal joint observable
when $\vg = c_1 \va + c_2 \vb .$
Note that $\vg$ should also satisfy the positivity conditions 
\eqref{eq:g1}--\eqref{eq:g4}
of $\Gg$,
whose shape in the $c_1$-$c_2$ plane strongly depends on 
the values of the inner products
$\va \cdot \va ,$ $\vb \cdot \vb ,$
and $\va \cdot \vb .$
}
\label{fig:1}
\end{figure}

The notation \eqref{eq:w} now takes the form
\[
	\begin{pmatrix}
		w (++) \\ w(+-) \\ w(-+) \\ w(--)
	\end{pmatrix}
	=
	\begin{pmatrix}
		c_1 +  c_2 + \gamma -2 \\
		c_1 -  c_2 + \gamma  \\
		-  c_1 +  c_2 + \gamma \\
		-   c_1  - c_2 + \gamma
	\end{pmatrix} 
\]
and the condition \eqref{eq:wmin} reduces to
\begin{equation}
	( \gamma < c_1 + c_2 < 2 - \gamma )
	\, \lor \,
	( - \gamma < c_1 - c_2 < \gamma) .
	\label{eq:wmin2}
\end{equation}
We also note that $0 < \gamma < 1 $ in this case.
In fact, if $\gamma = 0 ,$ we have $\Gg (++) = \Gg (--) =0$
which contradicts the linear independence of $\va$ and $\vb$.
If $\gamma =1$, we have $\Gg (+-) = \Gg (-+) = 0$,
which again contradicts the linear independence of $\va$ and $\vb .$
Thus we obtain the following corollary.

\begin{corollary}\label{coro:qubitmin}
Let $\Ea$ and $\Eb$ be compatible qubit observables and
suppose that $\va$ and $\vb$ are linearly independent.
Then their joint observable $\Gg$ is a minimal joint observable if and only if 
one of the following conditions holds.
\begin{enumerate}[1.]
\item \label{2:1}
$\vg,$ $\va ,$ and $\vb$ are linearly independent.
\item \label{2:2}
$\Gg$ satisfies either of the pairwise linear dependence conditions
\eqref{eq:dep1}, \eqref{eq:dep2}, \eqref{eq:dep5}, and \eqref{eq:dep6}.
\item \label{2:3}
$\vg$ can be written as $\vg = c_1 \va + c_2 \vb $
$(c_1 , c_2 \in \real)$
and satisfies \eqref{eq:wmin2}.
\end{enumerate}
\end{corollary}

For a fixed $\gamma \in (0,1) ,$
the region
for $\vg = c_1\va  + c_2 \vb$
corresponding to the conditions \ref{2:2} and \ref{2:3} of Corollary~\ref{coro:qubitmin}
is depicted in Fig.~\ref{fig:1}.

\section*{Acknowledgements}

T.H.\ acknowledges financial support from the Academy of Finland (Project No.\ 287750) and from the Academy of Finland Centre of Excellence program (Project No.\ 312058). 
Y.K.\ acknowledges financial support by the National Natural Science Foundation of China (Grants No.\ 11374375 and No.\ 11574405).

\newpage

\appendix 

\section{Order theoretic definitions}\label{sec:order}

Let $(X , \leq)$ be a partially ordered set (poset).
Any subset $Y$ of $X$ is also a poset.

A poset $X$ is \emph{totally ordered} if for any $x ,y \in X ,$ either $x \leq y$ or $y \leq x$ holds.
A totally ordered subset $C$ of $X$ is called a \emph{chain}.

\begin{itemize}
\item 
$X$ is \emph{upper directed} (or just \emph{directed})
$:\defarrow$ for any $x , y \in X $ there exists $z \in X$
such that $x \leq z$ and $y \leq z .$
\item 
$X$ is \emph{lower directed}
$:\defarrow$
for any $x , y \in X $ there exists $z \in X$
such that $x \geq z$ and $y \geq z .$
\item 
$X$ is \emph{upper} (resp.\ \emph{lower}) \emph{inductive}
$:\defarrow$
any chain of $X$ has an upper bound (resp.\ a lower bound).
\item
$X$ is \emph{upper} (resp.\ \emph{lower}) \emph{directed complete}
$:\defarrow$
every upper (resp.\ lower) directed subset of $X$
has a supremum (resp.\ an infinimum).
\end{itemize}
It is immediate from the above definition 
that an upper (resp.\ a lower) directed complete poset is upper (resp.\ lower) inductive.
According to Zorn's lemma, any upper (resp.\ lower) inductive poset has 
a maximal (resp.\ minimal) element.

Let $(I , \leq )$ be a directed set and let $X$ be a set.
A map $I \ni i \mapsto x_i \in X$ is called a net on $X .$
If $X$ is a poset and for any $i ,j \in I  ,$
$i \leq j$ implies $x_i \leq x_j$
(resp.\ $x_j \leq x_i$),
the net $(x_i)_{i \in I}$ is said to be monotonically increasing
(resp.\ decreasing).
A poset $X$ is upper (resp.\ lower) directed complete if and only if
the image of any monotonically increasing (resp.\ decreasing) 
net on $X$ has a supremum (resp.\ an infinimum).

\section{Pairwise linearly independent observables}\label{sec:pp2}

An observable $\Ao$ with an outcome set $\Omega$ is said to be \emph{pairwise linearly independent}
if any pair $(\Ao (x_1) , \Ao(x_2))$, $x_1 , x_2 \in \Omega ; \, x_1 \neq x_2$,
is linearly independent.
Every observable is post-processing equivalent to a pairwise linearly independent 
observable unique up to the permutation of the outcome set \cite{Kuramochi15b}.
An observable $\Ao$ is pairwise linearly independent if and only if $\Ao$ is minimal sufficient,
that is, for any Markov kernel $p \in \markov{\Omega}{\Omega}$ the condition $p \ast \Ao = \Ao$ implies 
$p (x, x^\prime) = \delta_{x , x^\prime} $ \cite{Kuramochi15b}.
We recall that two Markov kernels $p\in\markov{\Omega_1}{\Omega_2}$ and $q\in\markov{\Omega_2}{\Omega_3}$ can be combined into a new Markov kernel as follows:
\[
	(p \ast q) (x , x^\prime)
	:=
	\sum_{y \in \Omega_2}
	p (x , y) q(y , x^\prime) .
\]

The following proposition characterizes a Markov kernel 
that conserves the information of the post-processed observable.

\begin{proposition}\label{prop:mk}
Let $\Ao$ be a pairwise linearly independent observable with an outcome set $\Omega$
and let $p \in \markov{\Omega^\prime}{\Omega}$ be a Markov kernel
from $\Omega$ to a finite set $\Omega^\prime .$
Then the following conditions are equivalent.
\begin{enumerate}[(i)]
\item
$p\ast \Ao \ppsim \Ao .$
\item
For each $y \in \Omega^\prime$, the finite set
$\{p(y,x):x \in \Omega \}$ has at most one nonzero element.
\end{enumerate}
\end{proposition}

\begin{proof}
(i)$\implies$(ii).
Assume (i).
By assumption there exists a Markov kernel
$q \in \markov{\Omega}{\Omega^\prime}$ such that $(q\ast p) \ast \Ao =q \ast (p \ast \Ao) = \Ao$.
From the minimal sufficiency of $\Ao$ follows that $(q \ast p) (x , x^\prime) = \delta_{x , x^\prime}$. 
Now, assume that  there exist $y_0 \in \Omega^\prime $
and $x_1 , x_2 \in \Omega$ with $x_1 \neq x_2$
such that $p(y_0 ,x_1) p(y_0 ,x_2) \neq 0 .$
We can take $x_0 \in \Omega$ such that $q(x_0 , y_0) \neq 0.$
Then for $i =1,2 ,$
\begin{align*}
	\delta_{x_0 , x_i}
	&=
	\sum_{y \in \Omega^\prime}
	q(x_0 , y) p(y , x_i)
	\\
	&\geq
	q(x_0 , y_0) p(y_0 , x_i)
	\\
	&>0 .
\end{align*}
This implies $x_1 = x_0 =x_2 ,$
contradicting the assumption $x_1 \neq x_2 .$
Thus the condition~(ii) holds.

(ii)$\implies$(i).
Assume (ii).
Let 
$\Omega^\prime_1 := \left\{  y \in \Omega^\prime \mid (p \ast \Ao) (y) \neq 0 \right\} $
and let $\Bo$ be an observable with the outcome set $\Omega^\prime_1$, obtained by restricting  
$p \ast \Ao$ to $\Omega^\prime_1$.
Obviously, $\Bo$ is post-processing equivalent to $p \ast \Ao .$
By the assumption, for each $y \in \Omega^\prime_1$ there exists a unique element $x_y \in \Omega$ such that $p(y , x_y) \neq 0 .$
Then $\Bo (y) = p (y , x_y) \Ao (x_y) $ and hence
\[
	\Ao (x)
	=
	\sum_{y \in \Omega^\prime_1}
	\delta_{x , x_y} p (y , x_y) \Ao(x_y)
	=
	\sum_{y \in \Omega^\prime_1}
	\delta_{x , x_y} \Bo(y) ,
\]
which implies 
$\Ao \ppsim \Bo \ppsim p \ast \Ao .$
\end{proof}

\begin{corollary}
\label{coro:mk}
Let $\Ao$ be an observable with an outcome set $\Omega$ and let $p \in \markov{\Omega^\prime}{\Omega}$ be a Markov kernel
from $\Omega$ to a finite set $\Omega^\prime .$
Then $ p \ast \Ao \ppsim \Ao $ if and only if
for each $y \in \Omega^\prime $ and each $x_1 , x_2 \in \Omega ,$
if
$\Ao (x_1)$ and $\Ao (x_2)$ are linearly independent
then
$p(y , x_1) p (y ,x_2) = 0 .$
\end{corollary}
\begin{proof}
We take a pairwise linearly independent observable 
$\Bo \colon \Omega_0 \to \lh$ 
post-processing equivalent to $\Ao .$
Then there exist
$q \in \markov{\Omega}{\Omega_0}$
and $y \colon \Omega \to \Omega_0 $
such that
$\Ao (x) = q (x  , z (x)) \Bo (z(x))  $ $(x \in \Omega) ,$
$q(x , z^\prime) = 0$ if $z(x) \neq z^\prime ,$
and $\Ao = q \ast \Bo .$
Hence
$ p \ast \Ao = (p \ast q ) \ast \Bo ,$
and 
\begin{equation}
	p \ast q (y , z) = \sum_{x \in \Omega} p (y , x) q (x, z) 
	= \sum_{x \in \Omega , z(x) = z} p(y,x) q(x,z)
	\label{eq:pastq}
\end{equation}

Now assume $\Ao \ppsim p \ast \Ao .$
Then $\Bo \ppsim (p\ast q) \ast \Bo$
and Proposition~\ref{prop:mk} implies that for each $y \in \Omega^\prime$
$p \ast q (y,z) \neq 0 $
holds at most one $z \in \Omega_0 .$
If $\Ao(x_1)$ and $\Ao(x_2)$ are linearly independent, 
then
$z(x_1) \neq z(x_2) $
and $q(x_1 , z(x_1)) \neq 0 \neq q(x_2 , z(x_2)) .$
Hence \eqref{eq:pastq} implies $p(y,x_1) p(y,x_2) = 0.$
Conversely, if $\Ao$ is not post-processing equivalent to $p \ast \Ao ,$
then from $p \ast \Ao \ppsim (p \ast q) \ast \Bo $
and Proposition~\ref{prop:mk}, there exist $y \in \Omega^\prime$
and $z_1 , z_2 \in \Omega_0$
such that
$p\ast q (y,z_1) p\ast q (y,z_2) \neq 0$
and $z_1 \neq z_2 .$
From \eqref{eq:pastq} we can take $x_1 , x_2 \in \Omega$
such that
$z(x_1) = z_1 ,$
$z(x_2) = z_2 ,$
and
$p(y , x_1) q(x_1 , z_1) p (y,x_2) q(x_2,z_2) \neq 0.$
Then $\Ao (x_1)$ and $\Ao(x_2)$ are linearly independent
and $p(y,x_1) p(y,x_2) \neq 0 .$ 
\end{proof}

\section{Finite-dimensional polyhedra}\label{sec:polytope}

This appendix briefly describes some facts about
the finite-dimensional polyhedra used in the main part.

A convex set $K$ on $\real^n$ is called
a \emph{polyhedron} if there exist $m \geq 1 ,$
$\va^i \in \real^n ,$ and $\alpha^i \in \real$
$(i = 1 ,\dots , m)$
such that
\[
	K =
	\{
	\vx \in \real^n 
	\mid 
	 \va^i \cdot \vx \geq \alpha^i \, 
	(\forall  i \in \{  1, \dots , m\})
	\} ,
\]
i.e.\
$K$ is the set of solutions of a finite number of linear inequalities.
If a polyhedron $K$ is compact, $K$ is called a \emph{polytope}.
The set of extreme points $\mathrm{ex} (K)$ of the polytope $K$
is a finite set and can be calculated as follows.
Consider a subset $I\subseteq \{ 1 , \dots , m \}$
such that the system of linear equalities
\begin{equation}
	\va^i \cdot \vx = \alpha^i
	\quad 
	(i \in I)
	\label{eq:axa}
\end{equation}
has a unique solution.
Let $\mathcal{I}$ be the set of such subsets of $\{ 1 , \cdots , m\}$
and let $\vx_I$ be the unique solution of \eqref{eq:axa}
for each $I \in \mathcal{I} .$
Then $\mathrm{ex}(K)$ is given by
\[
	\mathrm{ex} (K)
	=
	\{
	\vx_I \mid
	I \in \mathcal{I} ,\, \vx_I \in K
	\} .
\]
(\cite{FDCO01}, Proposition~3.3.1).

A subset $C \subseteq \real^n$ is called a (pointed) \emph{polyhedral cone}
if there exist $m \geq 1 $ and
$\va^i \in \real^n$
$(i = 1 ,\dots , m)$
such that
\[
	C =
	\{
	\vx \in \real^n 
	\mid 
	 \va^i \cdot \vx \geq 0, 
	(\forall  i \in \{  1, \dots , m\})
	\} ,
\]
i.e.\
$C$ is the set of solutions of a finite number of \emph{homogeneous} linear inequalities.
For simplicity we assume that the linear span of $\{ \va^i\}_{i=1}^m$
coincides with the full space $\real^n .$
Then if $C \neq \{ 0\} ,$ 
$C$ is a conical hull of some finite set $A$ and $A$ can be calculated as follows.
Consider a subset $J \subseteq \{ 1 , \dots , m \}$
such that the rank of the linear equalities
\begin{equation}
	\va^i \cdot \vx = 0
	\quad 
	(i \in J)
	\label{eq:axac}
\end{equation}
is $n-1 $ and let $\mathcal{J}$ denote the set of such subsets.
For each $J \in \mathcal{J}$ we denote by $\Delta_J$
the line consisting of the solutions of \eqref{eq:axac}.
Then either $C\cap \Delta_J = \{ 0\}$
or 
$C\cap \Delta_J$ is a $1$-dimensional face of $C$
(\cite{FDCO01}, Proposition~3.3.2).
Let $\mathcal{J}^\prime := \{ J \in \mathcal{J} \mid   C \cap \Delta_J \neq \{0\}  \} .$
For each $J \in \mathcal{J}^\prime$ we take 
$\vy_J \in (C \cap \Delta_J) \setminus \{ 0 \}  .$
Then $A$ can be given by
$A = \{  \vy_J \}_{J \in \mathcal{J}^\prime} .$
If $\mathcal{J}^\prime = \emptyset ,$
we have $C = \{0\} .$

\end{document}